\documentclass[manuscript]{acmart} 

\usepackage{amsmath}
\usepackage{amsfonts}
\usepackage{bbold}

\setcopyright{acmcopyright}
\copyrightyear{2023}
\acmYear{2023}
\acmDOI{TBD}
\acmConference[RecSys 2023]{17th ACM Conference on Recommender Systems}{Sept 18--22,
  2023}{Singapore}

\begin{document}

\title{Of Spiky SVDs and Music Recommendation}

\author{Darius Afchar}
\email{research@deezer.com}
\author{Romain Hennequin}
\affiliation{%
  \institution{Deezer Research}
  \city{Paris}
  \country{France}
}
\author{Vincent Guigue}
\affiliation{%
  \institution{AgroParisTech, MLIA -- Sorbonne Universités}
  \country{France}
}

\begin{abstract}
The truncated singular value decomposition is a widely used methodology in music recommendation for direct similar-item retrieval and downstream tasks embedding musical items.
This paper investigates a curious effect that we show naturally occurring on many recommendation datasets: spiking formations in the embedding space. 
We first propose a metric to quantify this spiking organization's strength, then mathematically prove its origin tied to underlying communities of items of varying internal popularity.
With this new-found theoretical understanding, we finally open the topic with an industrial use case of estimating how music embeddings' top-k similar items will change over time under the addition of data.
\end{abstract}

\maketitle

\section{Introduction}

There is no unique definition of music recommendation, but rather a range of tasks that fall under this name: track sequence recommendation, context-aware recommendation, playlist continuation or generation, similar track, artist, or playlist retrieval \cite{schedl2018current}. These settings represent the many use cases of recommendation found in the wild (\textit{e.g.,} in a streaming service).
Despite this proteiformity, many recommenders leverage a model of similar item retrieval as a basis for their computation. A standard methodology is to rely on some high-dimensional representations of the musical items to recommend.
That is why we see more and more work aiming to find general-purpose representations of music tracks (\textit{e.g.,} self-supervision on music spectrograms \cite{wang2022towards}) to be used in various downstream recommendation scenarios. 
In the realm of collaborative filtering, the literature is dominated by matrix factorization techniques to create such representations.
Compared to other research fields of machine learning, it has been reported in recent years that modern approaches based on neural networks did not bring significant benefit over simpler baselines of factorization and also suffered many reproducibility issues \cite{rendle2019baseline,ferrari2019we,rendle2020neural,ferrari2021troubling}.
In this context, the truncated Singular Value Decomposition (SVD) is still a strong contender for factorization. It has demonstrated consistent performances across a variety of datasets \cite{cremonesi2010topn}. SVD also benefits from rich theoretical literature and has been implemented in efficient ways suitable to an industrial scale (\textit{e.g.,} distributed RSVD \cite{halko2011finding, constantine2011tall, briand2021semi}). That is why SVD is still widely used in music recommendation.

Nevertheless, despite not requiring much hyperparameter tuning, some blind spots still exist in using SVD as embeddings \cite{levy2015improving}.
For instance, after approximating a dataset $M$ with the low-rank decomposition $U\Sigma V^*$, should we use $U$, $U\Sigma$, or $U\Sigma^p$ as a basis of similarity? Should the dot product or cosine similarity be used?
Many practices for item retrieval rely on experimental results but have yet to find theoretical justifications.
In this paper, we submit the reader with a curious geometry that tends to emerge when applying SVD to multiple recommendation datasets: \textbf{spikes}. We observe that embedding vectors tend to self-organize along lines that pass through the origin (see Figure \ref{fig:overview}). This effect is understudied and goes unnoticed when using the cosine distance for similarity since the normalization squashes all points on a hypersphere \cite{levy2015improving}. However, we show that leveraging the norm of vectors may be more insightful than it first appears and that they should not be discarded.
In detail, we prove that spikes represent communities, and embeddings' norms stem from their varying importance within that community (\textit{i.e.,} intra-popularity). These two latter features are notably prevalent in music data \cite{cross2001music, jacobson2009musically} and substantiate the need for theoretical support in this context. As an opening to our results, we show that the norm is strongly indicative of the stability of embeddings, which has industrial implications on how music track representations evolve through time.
To our knowledge, this spiking behavior has yet to be leveraged for recommendation before.


The paper proceeds as follows: we first propose a unified metric to quantify this spiking effect, then we mathematically prove the equivalence of spiking geometry with the presence of communities and intra-popularity within each community. We discuss and open our formalization with a practical case of evaluating the stability of music embeddings.

All experiments may be reproduced via our code repository at \url{https://github.com/deezer/spiky_svd}.

\begin{figure}[h]
  \centering
  \includegraphics[width=0.9\linewidth]{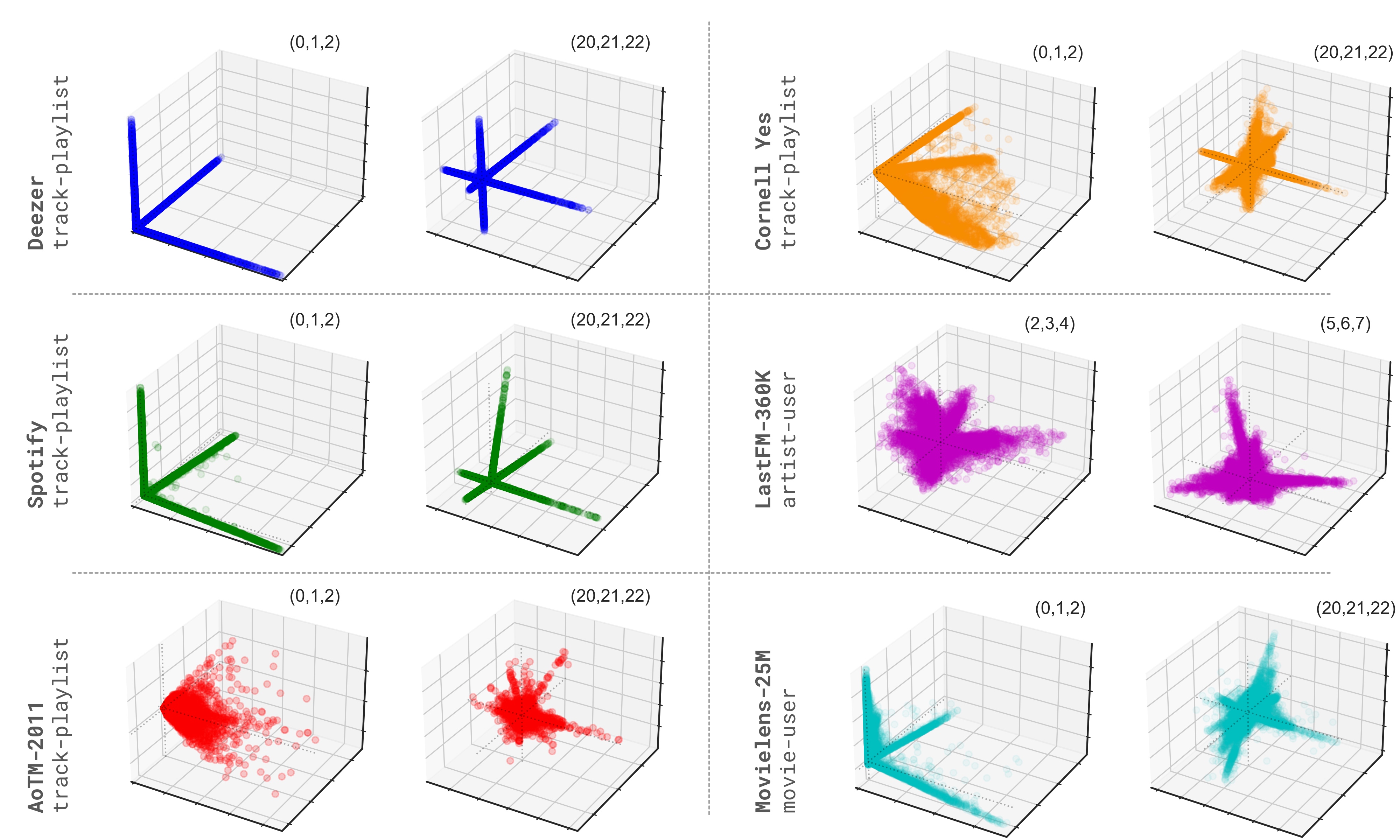}
  \caption{Naturally emerging "spikes" formations across several independently computed SVDs. Since the embedding vectors are high-dimensional, we can only display some projections of the points on slices of dimensions (indicated above).}
  \Description{Four subplots}
  \label{fig:overview}
  \vspace{-2em}
\end{figure}

\section{Geometry of the spikes}

We present and compute SVDs on several datasets and show the emergence of spike formations. We then propose establishing a measure of the "spikiness" of a distribution of embeddings, to quantify the strength of this effect.

\subsection{Background on the SVD}
\label{sec:background_svd}

We first provide the reader with a brief theoretical background. The \textit{Truncated Singular Value Decomposition} (SVD\footnote{A full-rank decomposition is, of course, intractable to use since it would create embeddings in millions of dimensions. So we abuse notation and drop the T in the acronym. As a reference, the truncated method can also be found under the designation \textit{pureSVD} in the recommenders community.}) is a factorization technique that aims to approximate a matrix $M \in \mathbb{R}^{n \times m}$ as the product of three matrices $U \in \mathbb{R}^{n \times f}$, $\Sigma \in \textrm{diag}(\mathbb{R}^f_+)$, and $V \in \mathbb{R}^{m \times f}$, such that $ \| M - \hat M \|_F^2 $ is minimal, with $\hat M = U \Sigma V^*$ and $U$ and $V$ are orthonormal by columns (\textit{i.e.,} $U^*U = V^*V = I_f$). $\hat M$ is a matrix of rank $f$, with $f$ chosen such that $f \ll \min(n,m)$. $\hat M$ is thus said to be a low-rank approximation of $M$. The "truncated" term refers to the optimal solution for $\hat M$ where only the $f$ largest singular values and eigenvectors of $M$ are retained for the decomposition (Eckart-Young-Mirsky theorem \cite{Eckart1936}).

In the context of music collaborative filtering, several remarks can be made. $M$ is often a large sparse matrix ($n$ and $m$ are frequently counted in millions) and often either denotes track-user, artist-user, or track-playlist interactions. $M$ is often positive, and sometimes additionally binary (\textit{e.g.,} with the formalism of implicit feedbacks \cite{rendle2009bpr}). Since the goal is often to leverage $\hat M$ for similar item retrieval rather than rebuilding $M$, the term $V$ is often ignored, and only the similarities between each row of $U$ or $U \Sigma$ are considered. The close-form solution to the SVD is such that $U$ is the $f$ dominant eigenvectors of $MM^*$. Therefore, it does not matter whether we compute the SVD on $M$ or $MM^*$ to retrieve $U$.
With an abuse of notations, we identify $M$ to its symmetrized version $MM^*$ since it does not impact the resulting $U$.

\subsection{Datasets}
\label{sec:datasets}

In order to demonstrate that spiking SVDs are not tied to a particular setting, we use six recommendation datasets: Spotify Million Playlist \cite{spotifyDataset} (Spotify), The Art of the Mix \cite{aotmDataset} (AoTM-2011), Cornell's playlists \cite{cornellDataset}, LastFM-360K \cite{lastfmDataset}. We also include the classic MovieLens-25M dataset \cite{movielensDataset} to hint at an extension of our result to other settings, and a private dataset from \textit{Deezer} containing twelve SVDs computed over one year in an industrial context.
This latter dataset will be further exploited in section \ref{sec:opening} to evaluate the stability of embedding representations.

We compute SVDs using the same setting for all six datasets. We use the symmetrized version $MM^*$ of the matrix. To curb popularity biases, denoise, and sparsify the matrix, we filter the top-$k$ highest non-null interactions per item. As commonly done in information retrieval, we apply a \textit{positive pointwise mutual information} (PPMI) normalization to the resulting matrix, which further helps with popularity biases and retrieval performances \cite{levy2015improving} and has been shown to be equivalent to modern skip-gram formulations with negative sampling \cite{levy2014neural}, which is common for music embeddings \cite{cheng2017exploiting, caselles2018word2vec, briand2021semi}. We choose $f = 128$ for the decomposition and take $E = U\Sigma$ as the resulting embeddings\footnote{Any other choice of normalization as $U\Sigma^p$ with $p \in [0, 1]$ simply scales the embedding space differently but does not impact the presence of spikes.}. The rows of $E = (e_1, ... e_n)$ provide a latent representation for each item.

The obtained SVDs are displayed in Figure \ref{fig:overview}. Since we cannot lay out all 128 dimensions in a single figure, we report curious readers to our repository for a detailed visualization. As a general comment, it seems that the spikiness of embeddings is stronger on dominant dimensions (\textit{i.e.,} eigenvectors associated with the higher singular values and hence the first indexes) and tends to degrade to a blurry cloud of points in the higher indexes. This seems to align with the literature on robust SVD \cite{stewart1998perturbation} that dominant components first reconstruct underlying structures and then reconstruct noisy perturbations.
The \textit{Deezer} and \textit{Spotify} datasets exhibit the cleanest spikes overall.

\subsection{Spikiness metric}
\label{sec:spikiness}

We could not find an off-the-shelf definition of spikes of points that worked in high-dimension\footnote{The closest applicable notion is probably the measure of \textit{kurtosis}. We have tried several experiments with this statistic. However, it turned out to be very sensitive to hyper-parameter choice and does not have a straightforward extension to multi-dimensional variables.}. In our context, it seems suited to study the embeddings with the highest norms as many candidate peaks for spikes. We can then measure whether many other points with a lower norm are collinear to them\footnote{Note that the spikes all seem to pass through the origin.}. If a small number of peaks are collinear to most of the rest of the distribution, it means that the distribution is spiky overall.

Since the distributions are noisy, we must instead rely on approximate collinearity. We set some thresholds: $\theta$ and $\rho$. We then denote by $e^*$ a given embedding with a high norm and say that any other vector $e$ belongs to the spike if $\cos(e^*, e) > \cos(\theta)$. We iterate for candidate peaks $e^*$ in descending order of norm until we have captured a ratio $\rho$ of points in the distribution. This can be seen as a greedy heuristic to the partial set cover problem\footnote{... which is notoriously NP-complete in high dimension and thus could not be solved exactly with our number of considered points.}.
The obtained number of spikes is divided by $n$ -- the total number of embeddings -- and denoted \texttt{Spk}. We have $0 < \texttt{Spk} \leq \rho$, where the lower bound occurs when all vectors form a single spike, and the upper bound means that no two points are collinear.

The results are given in Table \ref{tab:spikiness} for $\cos \theta = 0.9$ and $\rho = 50\%$. Though this measure is computed with a heuristic and constitutes an upper bound to possibly more optimal choices of peaks to cover the distribution, the obtained results seem consistent with our intuitions on the spikiness of each SVD.
The \textit{Deezer} and \textit{Spotify} come up as the most spiky. To fully spell out their obtained result, half of the millions of embedding vectors' directions -- that could be arbitrarily anything in 128 dimensions -- fall under one of 361 and 1395 spikes.
We also confirm that the embeddings are spikier in their first dimensions and then exhibit a sudden increase of \texttt{Spk} when the noise dominates in later dimensions.

\begin{table}
  \caption{Estimated \texttt{Spk} on several datasets: needed ratio of points that are peaks of spikes to capture $\rho = 50\%$ of the total distribution of $n$ embeddings, with $cos \theta = 0.9$. The parameter
  @$f$ denotes several choices of truncation for the SVD. We add a Gaussian sampling as a reference on the improbability of finding collinear vectors in $\mathbb{R}^f$ -- even with a $\theta$ angle tolerance.}
  \label{tab:spikiness}
  \begin{tabular}{c|cccccc|c}
    \toprule
    Dataset&Deezer&Spotify&AoTM-2011&Cornell&LastFM-360K&Movielens-25M&$\mathcal{N}(0,1)$\\
    \midrule
     \texttt{Spk}@128 & 0.04\%   & 0.14\% & 18.1\% & 4.55\% & 2.42\% & 7.09\% & 50\% {\color{gray}$= \rho$}\\
     \texttt{Spk}@64 & 150/$n$ {\color{gray}< 0.01\%} & 228/$n$ {\color{gray}< 0.03\%} & 11.4\% & 1.59\% & 0.49\% & 3.05\% & 50\% \\
     \texttt{Spk}@32 & 50/$n$ {\color{gray}< 0.01\%} & 55/$n$ {\color{gray}< 0.01\%} & 2.31\% & 0.49\% & 0.14\% & 1.12\% & 50\% \\
     \texttt{Spk}@16 & 23/$n$ {\color{gray}< 0.01\%} & 23/$n$ {\color{gray}< 0.01\%} & 0.93\% & 0.09\% & 0.06\% & 0.02\% & 49.4\% \\ 
  \bottomrule
\end{tabular}
\end{table}

\vspace{1em}
We have shown a systematic spiking effect happening over multiple computed SVDs. The result that a significantly low number of spikes may cluster half of the music tracks is highly reminiscent of \textit{spectral clustering} (but with spikes) and hints at the presence of structured groups in the data.

\section{Modeling}

We next demonstrate a formal equivalence between the spiking structure in the SVD-embedding space and the presence of communities of varying node degrees in the graph associated with $\hat M$. This enables us to get network properties from the embeddings without explicitly building a graph, for which we show a practical impact for recommendations.

\subsection{Preliminaries}

We define some notations and outline some basic properties to be used in the next section.

\subsubsection{Spikes in $M$}
\label{sec:notations_theorem}

As argued in section \ref{sec:background_svd}, we may assume $M$ to be symmetric since we only seek to retrieve similarities from the "left side" of $M$. We denote by $\textbb{S}^n(\mathbb{R}_+)$ the space of symmetric matrices of $\mathbb{R}^{n \times n}_+$. Applying the spectral theorem, $M$ is diagonalizable as $PDP^*$, with $P$ an orthonormal basis of eigenvectors. We rank the eigenvalues $(d_i)$ and corresponding vectors in decreasing order of absolute value. We then have two useful lemmas:

\begin{lemma}
\label{lemma:svd_diag}
Matrix $M$ has a truncated SVD given by $\hat M = U\Sigma V^*$ such that:
$$U = P_{1 .. f} \quad\quad \Sigma = |D_{1 .. f, 1 .. f}| = \textrm{diag}(\sigma_1, ... \sigma_f) \quad\quad V = P_{1 .. f} \textrm{sign}(D_{1 .. f, 1 .. f})$$
where $|.|$ denote the term-by-term absolute value and $\textrm{sign}()$ the sign function, since the SVD requires $\Sigma$ to be positive.
\end{lemma}
\begin{proof}
We identify a desired decomposition. 
Note that the unitary eigenvectors of $P$ are not uniquely defined, and thus, so does the decomposition.
\end{proof}

\begin{definition}[Spike formation]
A set of vectors $(e_1, ... e_k)$ is said to form a spike -- as studied in \ref{sec:datasets} -- if we may rewrite them as $(\alpha_1 \vec{s}, ... \alpha_k \vec{s})$ with $\vec s$ a unitary vector and $(\alpha_1, ... \alpha_k) \in \mathbb{R}_+^k$ their distance from the origin.
\end{definition}

\begin{lemma}
\label{lemma:spikes}
Let $M$ have a decomposition $\hat M = U\Sigma V^*$.
If the rows $(u_{i_1}, ... u_{i_k})$ of $U$ belong to a spike, then, we have that the rows $(i_1, ... i_k)$ of $U\Sigma^p$ and of $V$ also form a spike:
$$ u_i \Sigma^p = \alpha_i \left( \vec{s} \cdot \Sigma^p \right)  = \alpha_i \vec{s}' \quad \quad \quad v_i = u_i \cdot \textrm{sign}(D_{..f,..f}) = \alpha_i  \left(  \vec{s} \cdot \textrm{sign}(D_{..f,..f}) \right) = \alpha_i \vec{s}''$$
\end{lemma}
The spiking effect is insensitive to any choice of embedding as $U \Sigma^p$ or $V$: only the directions of the spikes change. In particular, the vector $\vec s''$ of $V$ remains unitary.
In the rest of this section, we arbitrarily fix the embeddings as $U \Sigma$.

\subsubsection{Community graph models} The Stochastic Block Model (SBM) is an important model introduced in \cite{holland1983stochastic} as a means to generate random graphs according to an underlying structure in communities. Given a set of $n$ nodes indexed as $\llbracket 1, n \rrbracket = \{1, 2, ... n \}$, a partition $ C: \llbracket 1, n \rrbracket \rightarrow \llbracket 1, K \rrbracket$ of the nodes into $K$ communities, and a matrix $B \in \mathcal{S}^K([0,1])$ of edge probabilities, the SBM model samples an adjacency matrix $A \in \mathcal{S}^n(\{0,1\})$ such that $A_{i,j} = A_{j,i} \sim \mathcal{B}(B_{C(i), C(j)})$, with $\mathcal{B}$ the Bernouilli law. On average, we thus obtain a matrix with a constant edge degree within and between any pair of communities of nodes.
If there is not a unique way of defining communities \cite{fortunato2010community, Lehmann2014}, it is usually agreed that the edge density within each community must be superior to that with other ones \cite{radicchi2004defining}, which translates as $B_{i,i} > B_{i, j}$ for all $i \neq j$. This last property is known as \textit{assortivity}. Music data has been shown to be particularly assortive \cite{jacobson2009musically}.

The \textit{Degree-Corrected SBM} (DCBM) \cite{karrer2011stochastic} introduces an additional set of parameters $(\alpha_i)_1^n \in [0, 1]^n$ allowing the edge probability to vary inside communities. Concisely, the DCBM proposes to sample edges between $i$ and $j$ with a probability $\alpha_i \alpha_j B_{C(i), C(j)}$. An example is depicted in Figure \ref{fig:dcbm}. In order to make the parameters identifiable, we may renormalize $B$ to such that $B_{i,i} = 1$. For further properties on the DCBM, we refer readers to \cite{karrer2011stochastic, lei2015consistency}.



\subsection{Equivalence of spikes and degree-varying communities}

We may finally spell out our main theorem linking recommendation embeddings and modeling of communities:

\begin{theorem}
\label{theorem:main} Let $M \in \mathcal{S}^n(\mathbb{R}_+)$ have a SVD $\hat M = U \Sigma V^*$ as defined in section \ref{sec:background_svd}.
Let $E = U \Sigma = (e_1, ... e_n)$ be embeddings that exhibit $K$ distinct spikes within its rows whose directions are given by $\vec s_1, ... \vec s_K$.
If row $e_i$ belongs to a spike, let us define $C(i)$ its corresponding spike index.
We additionally denote $n' \leq n$ the total number of points belonging to spikes. 
For ease of notation, we rearrange $E$ and $V$ such that the $n'$ first rows belong to spikes.
Then, the submatrix $\tilde M = | \hat{M}_{1..n', 1..n'} |$ is proportional to the mean value of a DCBM.

\end{theorem}

\begin{proof}
For all $i \in \llbracket 1, n' \rrbracket, e_i = \alpha_i \vec{s}_{C(i)}$. Similarly, $v_i = \alpha_i \vec{s}'_{C(i)}$, thanks to lemma \eqref{lemma:spikes}.
Then, for all $i,j \in \llbracket 1, n' \rrbracket^2$,
$$(\hat M)_{i,j} = (EV^*)_{i,j} = \alpha_i \alpha_j \langle \vec s_{C(i)}, \vec s'_{C(j)} \rangle
\quad \quad \quad
|\hat{M}_{1..n', 1..n'}| = K^2 \cdot \textrm{DCBM}\left( (C(i))_{i=1}^{n'} \, ; \, B \, ; \, \left( \frac{\alpha_i}{K} \right)_{i=1}^{n'} \right)$$
with $K = \max_{i} \alpha_i$, $ B \in [0,1]^{K \times K}$ the matrix $B_{k,l} = | \langle \vec s_k , \vec s'_l  \rangle |$, and $\frac{\alpha_i}{K} \in [0, 1]$ as required.
\end{proof}
In practice, there are good reasons to believe that $\hat M$ will be (roughly) non-negative and that the absolute value can be dropped in $\tilde M$. When $f=1$, the Perron-Frobenius theorem \cite{perron1907theorie} ensures the non-negativity, which is also trivially true when $f \geq \textrm{rank}(M)$. In between, we may expect negative values to principally arise from punctual approximation artifacts of null values of $M$, rather than large negative blocks having $\langle \vec s_k , \vec s'_l  \rangle < 0$.

\begin{figure}[h]
  \centering
  \includegraphics[width=\linewidth]{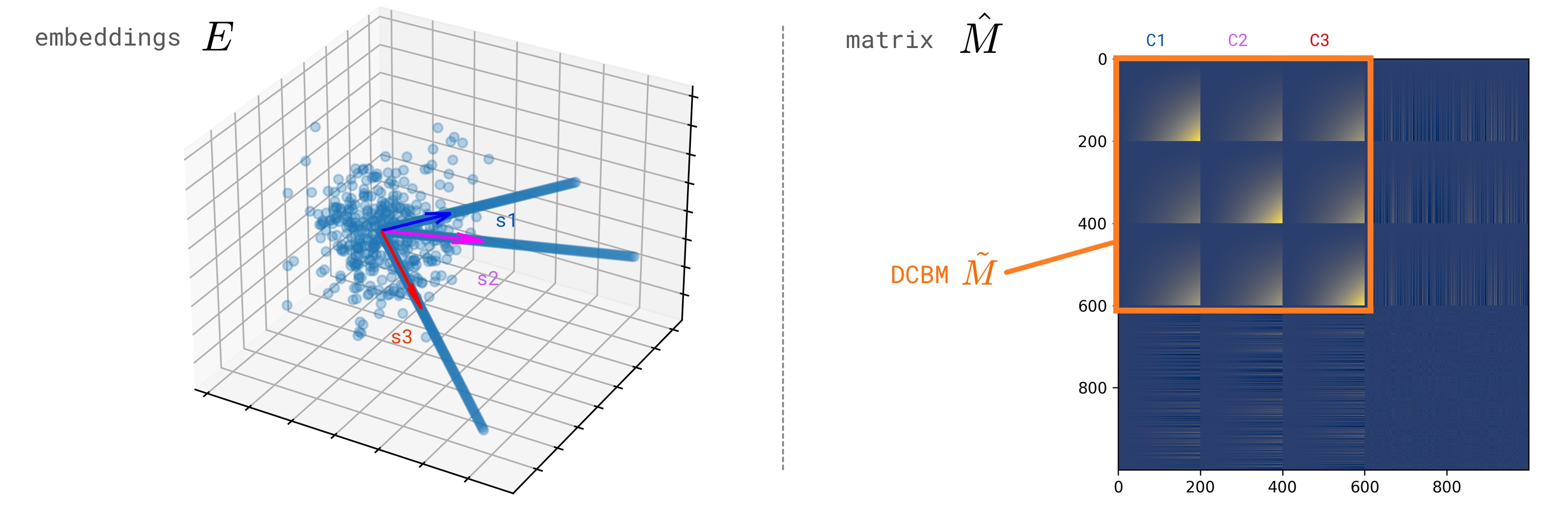}
  \caption{Illustration of theorem \eqref{theorem:main} with a synthetic matrix with $n = 1000$ and $n' = 3 \times 200$ points forming three spikes. We highlight the three direction vectors $\vec{s}_{\{1,2,3\}}$ in $E$ and their corresponding communities $C_{\{1,2,3\}}$ in $\hat M\,$ (reindexed for clarity).
  }
  \vspace{-0.2em}
  \Description{Two subplots}
  
  \label{fig:dcbm}
\end{figure}

\subsubsection*{Interpretation} We have shown that the sub-matrix $\hat{M}_{1..n', 1..n'}$ follows a structure in communities given by a DCBM, as illustrated in Figure \ref{fig:dcbm} with a synthetic example.
If we think geometrically, the DCBM is merely a matrix with constant blocks that is multiplied element-wise by the rank-1 matrix $[\alpha_1, ... \alpha_{n'}]^*[\alpha_1, ... \alpha_{n'}]$.
Our interpretation of this theorem is the following. Since $M$ is often sparse (\textit{c.f.,} section \ref{sec:datasets}), rebuilding as many communities as possible with rank-1 matrices is intuitively more optimal to minimize $\| M - \hat M\|_F^2$ than "spending" many eigenvectors on a single community to fine-tune it with a matrix of higher rank.
The widespread emergence of spike formations we uncovered across multiple datasets would thus stem from this general underlying prioritization of rank-1 approximation of communities. This coincidentally results in a DCBM with convenient theoretical properties that we can now use.
We leave the detection and analysis of eventual rank > 1 approximation out of the scope of this paper.

\subsubsection*{Related work} We found theorem \eqref{theorem:main} to be related to a result from community detection that the diagonalization of DCBM matrices leads to the appearance of spikes \cite{jin2012score, lei2015consistency}. Let us stress the difference. We prove the reciprocate (\textit{i.e.,} spikes $\Rightarrow$ DCBM), which is novel and in a more general case since a portion of $M$ is noisy ($n' < n$). In spirit also, the former literature aim to correct this effect -- seen as noise -- while we instead try to exploit it to gain knowledge on $\hat M$.

\subsection{Practical advice}
\label{sec:advice}

The duality of spike formations with DCBM communities has many convenient properties for recommendation. We highlight some corollaries of our theorem that provide practical advice when using embeddings computed with a SVD.

\subsubsection{Fast degree estimation}
As a first remark, the $(\alpha)_i$ parameters of our theorem's DCBM may be easily retrieved as the norm $\| e_i \|_2$ of embeddings that belong to a spike. Access to the degree of a node of the denoised matrix $\hat M$ without explicitly building a graph may open doors for many benefits. For instance, the topology of a graph may be estimated from its degree distribution (\textit{e.g.,} whether the graph is scale-free), which has been suggested to be indicative of a rather social or information nature of the music graphs \cite{cano2006topology}. Degree centrality has also been related to popularity biases in music data \cite{celma2008hits, south2020popularity} or used to create hierarchies \cite{afchar2023hierarchy}. More generally, the degree of a node is a strong predictor of its robustness to change in a graph \cite{gasteiger2019diffusion, geisler2020reliable}. In section \ref{sec:opening}, we will further exploit such properties of stability.

\subsubsection{Cosine similarity or dot product?}

The literature on community detection often seeks to curb the effect of the degree heterogeneity of the DCBM (\textit{e.g.,} \cite{karrer2011stochastic, qin2013regularized, gao2018dcbm}). One popular method is to normalize the eigenvectors with a spherical clustering \cite{lei2015consistency}. Normalizing on a hypersphere turns out to correspond to the widespread use of cosine similarity in recommendation settings.
It effectively removes the $(\alpha)_i$ terms, turning the DCBM into a simple SBM and thus attributing a unique value to all the items of a given community (the cosine similarity between their community direction $\vec s_i$ and that of the reference embedding to retrieve items from). From the recommendation lens, both unpopular and popular items will be reweighted with an equal relevance value.
This may be useful to improve \textit{discovery} objectives -- \textit{e.g.,} improving coverage and serendipity metrics \cite{ge2010beyond}, but conversely promotes filter bubbles by design \cite{pariser2011filter, zhang2012auralist}. 

On the flip side of the coin, measuring the internal popularity of a node within a community -- \textit{i.e.,} its "importance"  \cite{gao2018dcbm} -- may be helpful to navigate between communities while displaying their most representative items. This setting is, in particular, more suited for artist similarity and cold-start applications \cite{salha2019gravity, salha2021cold}. There, a dot product will always return the item with the highest norm of the community as top-similar, but may also merge popular items from nearby communities in the retrieved items.
The dot product thus seems more suited to propose a \textit{panopticon view} of the many available items, but with the adverse risk of being biased by popularity and thus making long-tail items harder to access.

As always, this choice of hyper-parameter depends on the task at hand and desired criteria to optimize.

\subsubsection{Tuning of $f$}

The \texttt{Spk} metric we proposed in \ref{sec:spikiness} may constitute a promising \textit{offline metric} on optimally choosing the truncation parameter $f$ in the SVD, \textit{i.e.,} finding when recommendation performances are saturated despite adding more embedding dimensions. In the light of theorem \eqref{theorem:main}, we can now reinterpret our metric as an approximation of $K$, where we have implicitly assumed that $\hat n' = \rho n$ points should be considered as being modeled by a DCBM. We have additionally used the empirical heuristic that noise was prevalent in the lower norm values of $E$. When the chosen $\hat n'$ comes close to the ground-truth $n'$, and the hypothesis starts to fail, noise starts to permeate, and the number of spikes grows linearly with the number of points (similar to the baseline Normal distribution where collinear points are rare).

In Table \ref{tab:spikiness}, this effect was striking between \textit{Deezer} and \textit{Spotify}, where the \texttt{Spk} metric was proportional between the two datasets for low values of $f$, then exploded at $f=128$ for the latter. We conclude that $f = 64$ for Spotify and $f=128$ for Deezer constitutes an appropriate choice to capture clean communities in the embeddings. 
This metric makes the assumption to treat non-spikes as uninformative noise, which dismisses potential higher-ranking approximations that may hold further information on the data.
However, we leave this last aspect for future work.


\section{Opening: industrial use-case}
\label{sec:opening}

We open our results with an industrial use case to illustrate the few previously highlighted properties.
Specifically, designing an industrial recommender system involves many challenges (\textit{e.g.,} scalability constraints) that can render adopting a complex machine learning model cumbersome \cite{afchar2022explainability}. Insights like the one we propose can thus be valuable to gain knowledge to better tune and evaluate recommenders leveraging simple item similarity models, such as the SVD, without changing the computation pipeline.
We had access to a history of size twelve of music track embeddings from \textit{Deezer}, spanning between May 2022 and May 2023, computed on user data with a truncated SVD as described in \ref{sec:datasets}. These SVDs contain 1.5 million common music tracks and have $f=128$.

\begin{figure}[h]
  \centering
  \includegraphics[width=\linewidth]{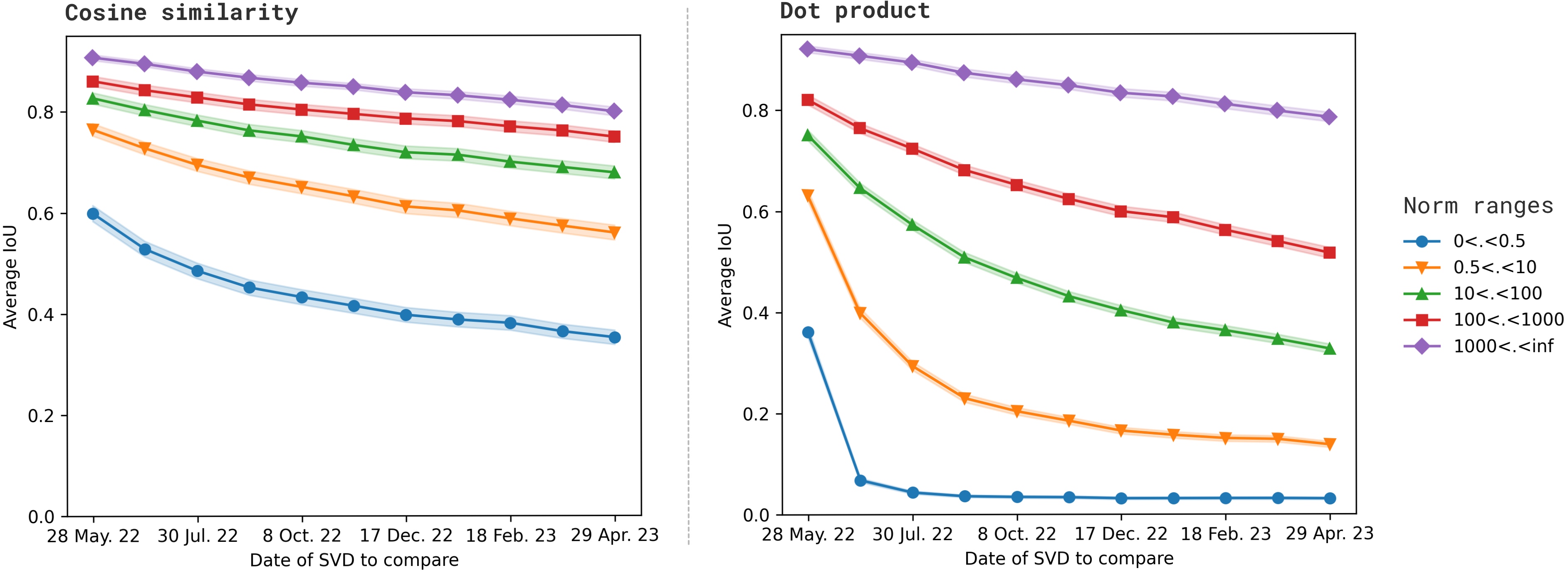}
  \caption{Stability through time of the top-$500$ of \textit{Deezer} embeddings starting in May 2022 and compared with eleven successive embeddings within five-week steps until May 2023. We condition the recommendation on a partition of the embedding norms (indicated on the right). Both the cosine similarity and dot product are considered. 95\% confidence intervals are displayed. }
  \Description{Two plots showing decreasing similarity curves.}
  \vspace{-1em}
  \label{fig:stability}
\end{figure}

Music data is particularly prone to exhibiting community effects -- linked to its many historical and cultural groundings \cite{cross2001music, jacobson2009musically}, but it is also ever-changing. New trends emerge every day (\textit{e.g.,} \textit{lo-fi} music, \textit{phonk}), others vanish (\textit{e.g.,} \textit{emo}, \textit{jumpstyle}), or are
sometimes revivified (\textit{e.g.,} fans of Kate Bush thanks to \textit{Stranger Things'} OST), or shape one of many types of evolution \cite{brodka2013ged}. It is thus natural to wonder how relevant some embeddings computed at a given date will remain in the following months under the change of the underlying data.

In section \ref{sec:advice}, we stressed that the estimated degree is considered a good predictor of the robustness of a node in a graph. We thus propose to link both notions and show that embeddings of higher norms have better stability in time than their lower-norm counterpart. Using the oldest SVD of May 2022 as a reference, we partition the embeddings equally into five according to their norm. We sample $1000$ points in each partition and compute their top-$500$ similar items in each SVDs -- using either the cosine similarity or the dot product -- among the almost 300.000 track candidates in each partition. We finally compute the mean Jaccard index (\texttt{IoU}) \cite{greene2010tracking} to quantify the relative similarity between the top items of the reference and subsequent compared date.

The results are given in Figure \ref{fig:stability}. We can see that conditioning our results on several norm ranges reveals strikingly different dynamics that could not have been inferred solely from a mean analysis. The dot product notably displays the most disparate curves where low-norm embeddings quickly become irrelevant while higher-norm ones remain over 80\% similar after almost a year.
Note that our point is not to state that change is necessarily undesirable in the embeddings, though we could expect well-established music genres to remain stable (\textit{e.g.,} jazz music tracks).
Instead, we envision this type of conditioned analysis as a means to better evaluate music recommenders by discriminating between various dynamics of embeddings -- and thus some expected associated properties.




\section{Conclusion}

Embeddings are commonly used in music recommendation. Yet, there are mostly so in a black-box manner, and many practices rely on empirical knowledge.
In this paper, we have proven that the widespread emergence of spikes in the embeddings of music items was tied to an underlying graph structure of degree-varying communities.
Our insights are particularly relevant for music data, where such structure is found naturally due to cultural and historical groundings.
Reframing a community detection model through the recommendation lens enabled us to formulate novel practical insights for music embeddings (\textit{e.g.,} cosine or dot product, stability) that are better grounded in theory.
We hope to see more work bridging the gap between graph theory and recommendation.
Our future work includes the analysis of rank~>~1 approximation effects on embeddings.

\clearpage
\bibliographystyle{ACM-Reference-Format}
\bibliography{main}

\end{document}